\documentclass{llncs}
\usepackage{todonotes}
\usepackage{xspace}
\usepackage{enumerate}
\usepackage{graphicx}
\usepackage{graphics}
\usepackage{amsmath,amssymb}

\usepackage{thm-restate}

\newcommand{\ortho}{\textsc{RepExt(ortho)}\xspace}
\newcommand{\topo}{\textsc{RepExt(top)}\xspace}
\newcommand{\port}{\textsc{RepExt(top+port)}\xspace}
\newcommand{\eps}{\varepsilon}

\begin{document}

\title{Extending Partial Orthogonal Drawings}

\author{Patrizio Angelini\inst{1} \orcidID{0000-0002-7602-1524} \and Ignaz Rutter\inst{2}\orcidID{0000-0002-3794-4406} \and Sandhya T P\inst{2} \orcidID{0000-0002-7745-3935}}
\authorrunning{P. Angelini \and I. Rutter \and Sandhya TP}

\institute{John Cabot University, Rome, Italy\\ \email{pangelini@johncabot.edu}
  \and Universität Passau, 94032 Passau, Germany\\
  \email{\{rutter, thekkumpad\}@fim.uni-passau.de}}

\date{}

\maketitle
\begin{abstract}
  We study the planar orthogonal drawing style within the framework of
  partial representation extension.  Let $(G,H,\Gamma_H)$ be a partial
  orthogonal drawing, i.e., $G$ is a graph, $H\subseteq G$ is a
  subgraph and $\Gamma_H$ is a planar orthogonal drawing of $H$. 
  
  We show that the existence of an orthogonal drawing~$\Gamma_G$ of $G$ that
  extends $\Gamma_H$ can be tested in linear time.  If such a drawing
  exists, then there also is one that uses $O(|V(H)|)$ bends per edge.
  On the other hand, we show that it is NP-complete to find an
  extension that minimizes the number of bends or has a fixed number
  of bends per edge.
\end{abstract}

\section{Introduction}

One of the most popular drawing styles are \emph{orthogonal drawings},
where vertices are represented by points and edges are represented by
chains of horizontal and vertical segments connecting their endpoints.
Such a drawing is \emph{planar} if no two edges share an interior point.  An
interior point of an edge where a horizontal and a vertical segment meet
is called a \emph{bend}.  The main aesthetic criterion for planar
orthogonal drawings is the number of bends on the edges.

A large body of literature is devoted to optimizing the number of
bends in planar orthogonal drawings.  The complexity of the problem
strongly depends on the particular input.  If the combinatorial
embedding can be chosen freely, then it is NP-complete to decide
whether there exists a drawing without bends~\cite{gt-ccurp-01}.  If
the input graph comes with a fixed combinatorial embedding, then a
bend-optimal drawing that preserves the given embedding can be
computed efficiently by a classical result of
Tamassia~\cite{t-emn-87}.  A recent trend has been to investigate
under which conditions the variable-embedding case becomes tractable.
For \mbox{maxdeg-3} graphs a bend-optimal drawing can be computed
efficiently~\cite{dlv-sood-98}, which has recently been improved to
linear time~\cite{dlop-oodp3-20}.  The problem is also FPT with
respect to the number of degree-4 vertices~\cite{dl-codve-98}, and if
one discounts the first bend on each edge, an optimal solution can be
computed even for individual convex cost functions on the
edges~\cite{br-oogdc-16,blr-ogdie-16}.  We refer to the
survey~\cite{dg-popda-13} for further references.

In light of this popularity and the existence of a strongly developed
theory, it is surprising that the planar orthogonal drawings have not
been investigated within the framework of partial representation
extension.  Especially so, since it has been considered in the related
context of simultaneous representations~\cite{accdd-sop-16}.

In the partial representation extension problem, the input graph $G$
comes together with a subgraph $H \subseteq G$ and a representation
(drawing)~$\Gamma_H$ of $H$.  One then seeks a drawing~$\Gamma_G$ of
$G$ that \emph{extends}~$\Gamma_H$, i.e., whose restriction to $H$
coincides with $\Gamma_H$.  The partial representation extension
problem has recently been considered for a large variety of different
types of representations.  For planar straight-line drawings, it is
NP-complete~\cite{p-epsld-06}, whereas for topological drawings there
exists a linear-time algorithm~\cite{adfjkpr-tppeg-15} as well as a
characterization via forbidden substructures~\cite{jkr-kttpp-13}.
Moreover, it is known that, if a topological drawing extension exists,
then it can be drawn with polygonal curves such that each edge has a
number of bends that is linear in the complexity
of~$\Gamma_H$~\cite{partial}.  Here the complexity of
$\Gamma_H$ is the number of vertices and bends in~$\Gamma_H$.  Most
recently the problem has been investigated in the context of
1-planarity~\cite{eghkn-ep1pd-20}.  Besides classical drawing styles,
it has also been studied for  contact representations~\cite{contactrep} and for geometric intersection representations,
e.g., for (proper/unit) interval
graphs~\cite{kkosv-eprig-17,kkors-eprpu-17}, chordal
graphs~\cite{kkos-eprsc-15}, circle graphs~\cite{cfk-eprcg-19}, and
trapezoid graphs~\cite{kw-eprtg-17}.

\begin{figure}[tb]
  \centering
  \includegraphics[page=2]{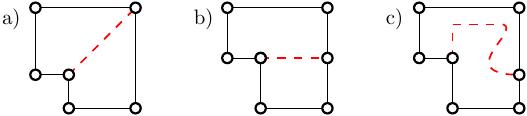}
%
    \caption{An instance of the partial representation extension problem $(G,H,\Gamma_H)$ is given. The graph $H$ is
    solid black and the edges of $E(G) \setminus E(H)$ are dashed red. (a) $(G,H,\Gamma_H)$ admits a
    planar extension, but not an orthogonal extension.  (b) $(G,H,\Gamma_H)$ admits an orthogonal extension with no bends (c) An orthogonal representation of $G$ (the curved part of the dashed
    edge has no bends) that extends the description of the solid black
    drawing of $H$.  There exists no drawing of $G$ with this
    representation that extends the given drawing of $H$.}
\label{fig:examples}
\end{figure}

In this paper, we provide an in-depth study of partial representation
extension problems for the orthogonal drawing style.  Since the
aesthetics are of particular importance for the quality of such a
drawing, we put a major emphasis on extension questions in relation to
the number of bends.  It is worth noting that even the seminal work of
Tamassia~\cite{t-emn-87} already mentions the idea of preserving the
shape of a given subgraph by maintaining its orthogonal representation
via modifications in his flow network.  However, this approach only
preserves the shape of the subgraph as described by an orthogonal
representation, and not necessarily its drawing.
Fig.~\ref{fig:examples} shows that there are partial planar
orthogonal drawings that can be extended in a planar way, but not
orthogonally (Fig.~\ref{fig:examples}a) and that, even if an
orthogonal representation $O_G$ of $G$ preserves a given orthogonal
representation $O_H$ of a drawing~$\Gamma_H$ of $H$, there does not
necessarily exist a drawing~$\Gamma_G$ of $G$ realizing $O_G$ that
extends~$\Gamma_H$~(Fig.~\ref{fig:examples}b).

\paragraph{Contribution and Outline.}

After presenting preliminaries in Section~\ref{sec:preliminaries}, 
we give a linear-time algorithm for deciding the existence of an
orthogonal drawing extension in Section~\ref{sec:test-extens}.
Then, we consider the realizability problem, where we are given an
orthogonal extension in the form of a suitable planar embedding, and
we seek an orthogonal drawing extension that optimizes the number of bends.  Along
the lines of a result by Chan et al.~\cite{partial}, we show
that there always exists an orthogonal drawing extension such that
each edge has a number of bends that is linear in the complexity of
$\Gamma_H$ in Section~\ref{sec:realizability}.  We complement these
findings in Section~\ref{sec:bend-minim-extens} by showing that it is
NP-hard to minimize the number of bends and NP-complete to test
whether there exists an orthogonal drawing extension with a fixed number of bends
per edge.  For proofs of the results marked with a \textnormal {[$\boldsymbol {\ast} $]}, please refer to the Appendix.

\section{Preliminaries}
\label{sec:preliminaries}

We call the circular clockwise ordering of the edges around a vertex $v$ in an embedding the \emph{rotation} at $v$. Let $G=(V,E)$ be a simple undirected graph and let~$H \subseteq G$ be
a subgraph.  We refer to the vertices and edges of $H$ as
\emph{$H$-vertices} and \emph{$H$-edges}, respectively.  Similarly, we
refer to the vertices of $V(G) \setminus V(H)$ and to the edges of
$E(G) \setminus E(H)$ as \emph{$G$-vertices} and \emph{$G$-edges},
respectively.

Let $( G, H, \Gamma_H)$ be a triple composed of a graph $G$, a subgraph $H \subseteq G$, and an orthogonal drawing $\Gamma_H$ of $H$. We denote by \ortho (\textsc{RepExt} stands for representation extension) the problem of testing whether $G$ admits an orthogonal drawing $\Gamma_G$ that extends $\Gamma_H$.
In $\Gamma_H$, we say that an $H$-edge is \emph{attached to} one of the four \emph{ports} of its end vertices. If there is no $H$-edge attached to a port of a vertex, then this port is \emph{free}; note that the free ports are those at which the $G$-edges can be attached in $\Gamma_G$. 
For two edges $e$ and $e'$ that are consecutive in the rotation at a vertex $v$ in $\Gamma_H$, we denote by $\mathcal P_H(e,e')=k$ the fact that there exist exactly $k$ free ports of $v$ when moving from $e$ to $e'$ in clockwise order around their common endvertex. We call $\mathcal P_H(e,e')=k$ a \emph{port constraint}, and we denote by $\mathcal P_H$ the set of all port constraints in $\Gamma_H$. Note that, for a vertex $v$ with rotation $e_1,\dots,e_h$ in $\Gamma_H$, with $h \leq 4$, we have $\sum_{i=1}^{h} \mathcal P_H(e_i,e_{i+1}) = 4 - \deg(v)$ (defining $e_{h+1} := e_1$). 


We now show that 
to solve an instance $(G,H,\Gamma_H)$ of the \ortho problem, it suffices to only consider the port
constraints determined by $\Gamma_H$ together with the embedding $\mathcal E_H$ of $H$ in $\Gamma_H$. More specifically, we prove the following characterization, which could also be deduced from~\cite{accdd-sop-16}.

\begin{restatable}[$\star$]{theorem}{characterization}
	\label{thm:characterization}
	Let $(G, H, \Gamma_H)$ be an instance of \ortho. Let $\mathcal E_H$ be the embedding of $H$ in $\Gamma_H$, and let $\mathcal P_H$ be the port constraints induced by $\Gamma_H$. Then, $(G, H, \Gamma_H)$ admits an orthogonal drawing extension if and only if $G$ admits a planar embedding $\mathcal E_G$ that extends $\mathcal E_H$ and such that, for every port constraint $\mathcal P_H(e,e')=k$, there exist at most $k$ $G$-edges between $e$ and $e'$ in the rotation at $v$ in $\mathcal E_G$, where $v$ is the common vertex of the $H$-edges $e$ and $e'$.
\end{restatable}

In view of Theorem~\ref{thm:characterization}, we define a new problem, called \port, which is linear-time equivalent to \ortho. An instance of this problem is a 4-tuple $( G, H, \mathcal{E}_H, \mathcal{P}_H )$ and the goal is to test whether $G$ admits an embedding $\mathcal E_G$ that satisfies the conditions of Theorem~\ref{thm:characterization}. In order to unify the terminology, we also refer to the \emph{Partially Embedded Planarity} problem studied in~\cite{adfjkpr-tppeg-15} as \topo (\textsc{top} stands for topological drawing). Recall that an instance of this problem is a triple $\langle G, H, \mathcal{E}_H \rangle$, and the goal is to test whether $G$ admits an embedding $\mathcal E_G$ that extends $\mathcal E_H$. As proved in~\cite{adfjkpr-tppeg-15}, \topo can be solved in linear time. 

\section{Testing Algorithm}
\label{sec:test-extens}

In this section we show that \ortho can be solved in linear time. By Theorem~\ref{thm:characterization}, it suffices to prove that \port can be solved in linear time. The algorithm is based on constructing in linear time, starting from an instance $(G,H,\mathcal E_H, \mathcal P_H)$ of \port, an instance $(G',H',\mathcal E_{H'})$ of \topo that admits a solution if and only if $(G,H,\mathcal E_H, \mathcal P)$ does.

In order to construct the instance $(G',H',\mathcal E_{H'})$ of \port, we initialize $G'=G$, $H'=H$, and $\mathcal E_{H'} = \mathcal E_H$. 
Then, for each vertex $v$ such that $ 1 < \deg_H(v) < \deg_G(v)$, we perform the following modifications; see Fig.~\ref{fig:replacement}.

Case 1: Suppose first that $\deg_H(v) = 3$ and $\deg_G(v) = 4$, and let $e=vw$ be the unique $G$-edge 
incident to $v$; refer to Fig.~\ref{fig:replacement}(a). Since $\deg_H(v) = 3$, there exist exactly two $H$-edges $e_1$ and $e_2$ such 
that $e_1$ immediately precedes $e_2$ in the rotation at $v$ in $\mathcal E_H$ and 
$\mathcal P(e_1,e_2)=1$.
Note that, to respect the port constraint, we have to guarantee that $e$ is placed between $e_1$ and $e_2$
in the rotation at $v$ in $\mathcal E_G$. 
For this, we subdivide $e$ with a new vertex $w'$, that is, we remove $e$ from $G'$, and we add the vertex $w'$ and the
edges $vw'$ and $w'w$ to $G'$. Also, we add $w'$ and $vw'$ to $H'$, and insert $vw'$ between $e_1$ and $e_2$
in the rotation at $v$ in $\mathcal E_{H'}$.

Case 2: Suppose now that $\deg_H(v) = 2$ and $\deg_G(v) \geq 3$. Let $e_1$ and $e_2$ be the two
$H$-edges incident to $v$, and let $e=vw$ and $e^*=vz$ be the at most two $G$-edges incident to $v$. 
We distinguish two cases, based on whether $\mathcal P(e_1,e_2)=2$ and $\mathcal P(e_2,e_1)=0$ (or vice versa), or $\mathcal P(e_1,e_2)=\mathcal P(e_2,e_1)=1$.

Case 2.a: If $\mathcal P(e_1,e_2)=2$, then we need to guarantee that both $e$ and $e^*$ (if it exists) are placed between $e_1$ and $e_2$
in the rotation at $v$ in $\mathcal E_G$; refer to Fig.~\ref{fig:replacement}(b). For this, we remove $e$ and $e^*$ from $G'$, and we add a new vertex $w'$ and the
edges $vw'$, $w'w$, and $w'z$ to $G'$. Also, we add $w'$ and $vw'$ to $H'$, and insert $vw'$ between $e_1$ and $e_2$
in the rotation at $v$ in $\mathcal E_{H'}$. Note that, if $e^*$ does not exist, this is the same procedure as in the previous case.
Case 2.b: If $\mathcal P(e_1,e_2)=\mathcal P(e_1,e_2)=1$, then we need to guarantee that $e$ and $e^*$ (if it exists) appear on different sides
of the path composed of the edges $e_1$ and $e_2$; refer to Fig.~\ref{fig:replacement}(c). Note that, if $e^*$ does not exist, then $e$ can be on any of the two sides of this path, and thus in this case we do not perform any modification. If $e^*$ exists, we subdivide $e$, $e^*$, $e_1$, and $e_2$ with a new vertex each, that is, we remove these edges from $G'$ ($e_1$ and $e_2$ also from $H'$), and we add four new vertices $w'$, $z'$, $w_1'$, and $w_2'$. Also, we add to $G'$ the edges $vw'$, $vz'$, $vw_1'$, and $vw_2'$, and the edges $w'w$, $z'z$, $w_1'w_1$, and $w_2'w_2$, where $w_1$ and $w_2$ are the endpoints of $e_1$ and $e_2$, respectively, different from $v$. Further, we add the edges $w'w_1'$, $w_1'z'$, $z'w_2'$, and $w_2'w'$ to $G'$. Finally, we add the edges $vw_1'$, $w_1'w_1$, $vw_2'$, and $w_2'w_2$ also to $H'$; in $\mathcal E_{H'}$, we place $w_1'w_1$ and $w_2'w_2$ in the rotations at $w_1$ and at $w_2$, respectively, in the same position as $e_1$ and $e_2$, respectively, in $\mathcal E_H$. The rotations at $v$, $w'$, $z'$, $w_1'$, and $w_2'$ in $\mathcal E_{H'}$ do not need to be set, since each of these vertices has at most two incident $H'$-edges. The above construction leads to the following lemma, whose full proof is in the Appendix.

\begin{figure}[tb]
	\centering
	\includegraphics{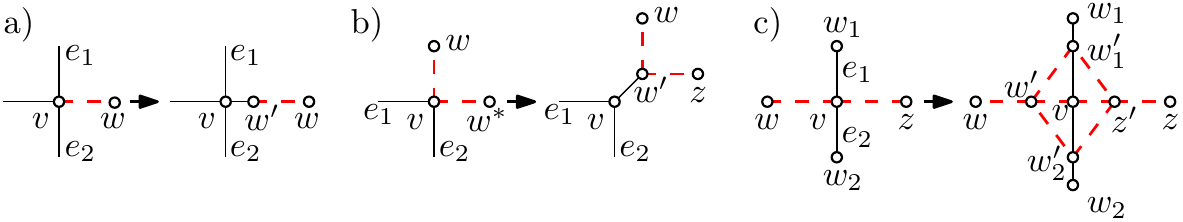}
	\caption{Gadgets for $H$-vertices}
	\label{fig:replacement}
\end{figure}

\begin{restatable}[$\star$]{lemma}{reduction}
\label{lem:reduction}
	The instance $(G',H',\mathcal E_{H'})$ has an embedding extension if
	and only if $(G,H,\mathcal E_H,\mathcal P)$ has an embedding extension 
	satisfying the port constraints.
\end{restatable}

\begin{theorem}\label{th:port-linear}
	The \port problem can be solved in linear time.
\end{theorem}

\begin{proof}
	Given an instance $I=(G,H,\mathcal E_H,\mathcal P)$ of \port, we construct the instance $I'=(G',H',\mathcal E_{H'})$ of
	\topo that has linear size as described above. This takes $O(1)$ time per vertex, and hence total linear time.
	By Lemma~\ref{lem:reduction}, $I$ has a solution if and only if $I'$ has one. Since the existence of a solution of $I'$ can be tested in linear time~\cite{adfjkpr-tppeg-15}, the statement follows.
\end{proof}
\noindent As a consequence of Theorems~\ref{thm:characterization} and~\ref{th:port-linear}, we conclude the following.
\begin{theorem}\label{th:ortho-linear}
	The \ortho problem can be solved in linear time.
\end{theorem}

\section{Realizability with Bounded Number of Bends}
\label{sec:realizability}

In this section we prove that, if there exists an orthogonal drawing extension for an instance $(G, H, \Gamma_H)$ of \ortho, then there also exists one in which the number of bends per edge is linear in the complexity of the drawing~$\Gamma_H$.  By subdividing $H$ at the bends of~$\Gamma_H$, we can assume that~$\Gamma_H$ is a bend-free drawing of~$H$.  To achieve the desired edge complexity, it then suffices to show that $O(|V(H)|)$ bends per edge suffice.
This result can be considered as the counterpart for the orthogonal setting of the one by Chan et al.~\cite{partial} for the polyline setting. In their work, in fact, they show that a positive instance $(G, H, \Gamma_H)$ of the \topo problem can always be realized with at most $O(|V(H)|)$ bends per edge when~$\Gamma_H$ is a planar straight-line drawing of~$H$. 

Our approach follows the algorithm given in~\cite{partial}, with a main technical difference which is due to the peculiar properties of orthogonal drawings.  Their algorithm first constructs a planar supergraph $G'$ of $G$ that is Hamiltonian using a method of Pach et al.~\cite[Lemma 5]{pach}.  The main step of the algorithm of Chan et al.~\cite{partial} involves the \emph{contraction} of some edges of $G^\prime$~\cite[Lemma 3]{partial}). This operation identifies the two end-vertices of the contracted edge and merges their adjacency lists. However, both the construction of the supergraph $G'$ and the contractions may produce vertices of degree greater than $4$, which implies that the resulting graph does not admit an orthogonal drawing any longer. As such, these operations are not suitable for the realization of orthogonal drawings. In order to overcome this problem, we consider instead the \emph{Kandinsky} model~\cite{kaufmann}, which extends the orthogonal drawing model to also allow for vertices of large degree. Once the drawing has been computed, we remove the previously added parts and by adding a small amount of additional bends on the $G$-edges, we arrive at a orthogonal drawing of the initial graph $G$. More specifically, we prove the following theorem:
\begin{restatable}[$\star$]{theorem}{drawing}
	Let $(G, H, \Gamma_H)$ be an instance of \ortho. Suppose that $G$ admits an orthogonal drawing $\Gamma_G$ that extends $\Gamma_H$, and let $\mathcal E_G$ be the embedding of $G$ in $\Gamma_G$. Then we can construct a planar Kandinsky drawing of $G$ in $O(n^2)$-time, where $n$ is the number of vertices of $G$, that realizes $\mathcal{E}_G$, extends $\mathcal{H}$, and has at most $262|V(H)|$ bends per edge. 
	\label{mainthm}
\end{restatable}
An overview of the algorithm to construct the desired Kandinsky orthogonal drawing $\Gamma^*_G$ of $G$, whose main steps follow the method in \cite{partial}, is given below.


\begin{enumerate}[Step 1:]
\item Consider a face $F$ of $\Gamma_{H}$ with facial walks $W_1,W_2,\ldots, W_k$. Construct an $\varepsilon$-approximation of $F$ and let $W_i'$ be the orthogonal polygon that approximates $W_i$, $1 \leq i \leq k$. Let $F'$ be the face bounded by the approximated boundary components of $F$; refer to Lemma~\ref{nest-approx}, and to Fig.~\ref{fig:epsilon} in the Appendix.

\item Partition $F^\prime$ into rectangles~\cite{rect} and construct a graph $K$ by placing a vertex at the center of each rectangle and by joining the vertices of adjacent rectangles.  Let $T$ be a spanning tree of $K$. For each facial walk $W_i$, add a new vertex near to $W_i$ as a leaf of $T$ (see Fig.~\ref{fig:rect}). 
\begin{figure}[tb]
\begin{center}
\includegraphics[page=4]{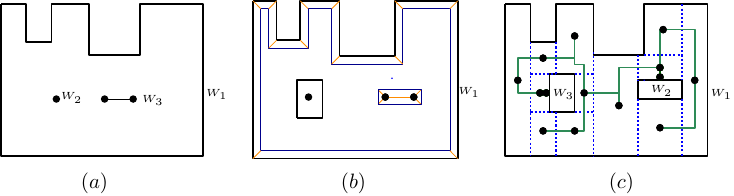}
\caption{(a) A face with outer walk $W_1$ and, inner facial walks $W_2$ and $W_3$. (b) An approximation $F^\prime$ of $F$. (c) A face and  a corresponding tree $T$ }
\label{fig:rect}
\end{center}
\end{figure}\item Construct the multigraph $G_F$ induced by the vertices lying inside or on the boundary of $F$ and by contracting each facial walk of $F$ to a single vertex. Then draw $G_F$ along $T$. Now, reconstruct the edges of $G\setminus H$ and the edges between $G_F$ and other components of $G$ inside $F$. Refer to Lemma~\ref{lemma3}  in the Appendix and to Fig.~\ref{lemma2figsim}.

\end{enumerate}
\begin{figure}[tb]
\begin{center}
\includegraphics{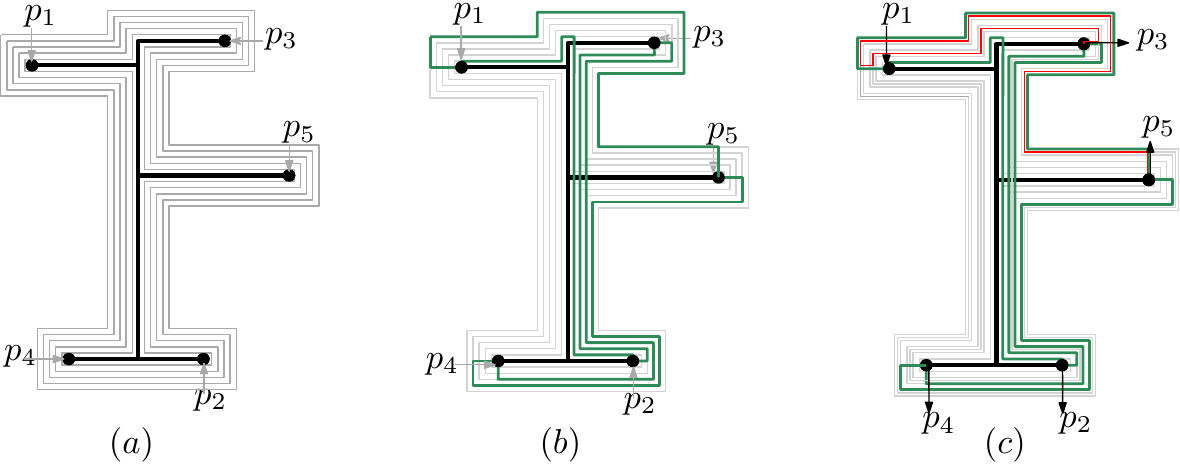}
\caption{(a) An orthogonal drawing of a tree $T$ together with approximations along $T$ (b) An orthogonal drawing of the Hamiltonian cycle $C$ with respect to $T$ (c) The edge $p_3p_5$ is drawn using approximations of $T$}
\label{lemma2figsim}
\end{center}
\end{figure}
\noindent We then transform $\Gamma^*_G$ into an orthogonal drawing $\Gamma_G$ of $G$ with $O(|V(H)|)$ bends per edge that extends~$\Gamma_H$. An illustration is given in Fig. \ref{fig:orthplnr}.

\begin{figure}[tb]
\begin{center}
\includegraphics{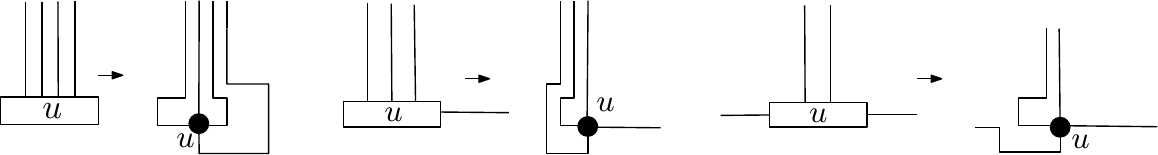}
\caption{Re-routing the edges incident to a vertex $u$ in the Kandinsky drawing $\Gamma^K_G$ to obtain the orthogonal drawing $\Gamma_G$.}
\label{fig:orthplnr}
\end{center}
\end{figure}
\begin{restatable}[$\star$]{theorem}{orthodrawing}
Let $(G, H, \Gamma_H)$ be an instance of \ortho. Suppose that $G$ admits an orthogonal drawing $\Gamma_G$ that extends $\Gamma_H$, and let $\mathcal E_G$ be the embedding of $G$ in $\Gamma_G$. Then we can construct a planar orthogonal drawing of $G$ in $O(n^{2})$-time, where $n$ is the number of vertices of $G$, that realizes $\mathcal{E}_G$, extends $\mathcal{H}$, and has at most $270|V(H)|$ bends per edge. 
\label{orthoplanar}
\end{restatable}
 

\section{Bend-Optimal Extension}
\label{sec:bend-minim-extens}


In this section we study the problem of computing an orthogonal drawing 
extension of an instance $I=(G,H,\Gamma_H)$ of \ortho 
with the minimum number of bends. Observe that, if $H$ is empty, this is
equivalent to computing a bend-minimal drawing of $G$, which is
NP-complete if the embedding of $G$ is not fixed. We
thus assume that $G$ comes with a fixed planar embedding $\mathcal E_G$
that satisfies the port constraints of $\Gamma_H$, and we study the
complexity of computing a bend-optimal drawing~$\Gamma_G$ of $G$ with 
embedding $\mathcal E_G$ that extends~$\Gamma_H$.

Here, we specifically focus on the restricted case where $V(H) = V(G)$
and $E(H) = \emptyset$, which we call \emph{orthogonal point set
  embedding with fixed mapping}.  We show that, even in this case, it
is NP-hard to minimize the number of bends on the edges.  On the
positive side, we show that in this case the existence of a drawing
that uses one bend per edge can be tested in polynomial time.

\begin{theorem}
  \label{thm:bend-optimal-hardness}
  Given an instance $(G,H,\Gamma_H)$ of \ortho, a planar embedding~$\mathcal E_G$ of $G$ that
  satisfies the port constraints of~$\Gamma_H$, and a number
  $k\in \mathbb N_0$, it is NP-complete to decide whether $G$ admits
  an orthogonal drawing~$\Gamma_G$ with embedding~$\mathcal E_G$ that
  extends~$H$ and has at most $k$ bends.  This holds even if
  $V(G) \setminus V(H) = \emptyset$, $E(H) = \emptyset$, and $E(G)$ is
  a matching.
\end{theorem}

\begin{proof}
  We give a reduction from the NP-complete problem \emph{monotone planar
  3-SAT}~\cite{dk-obsps-12}.  In this variant of 3-SAT, the variable--clause graph
  is planar and has a layout where the variables are
  represented by horizontal segments on the $x$-axis, the clauses
  by horizontal segments above and below the $x$-axis, and
  each variable is connected to each clause containing it by a
  vertical segment, the clauses above the $x$-axis
contain only positive literals and the clauses below contain only negative literals; see Fig.~\ref{fig:formula}a.

  \begin{figure}[t]
    \centering
    \includegraphics{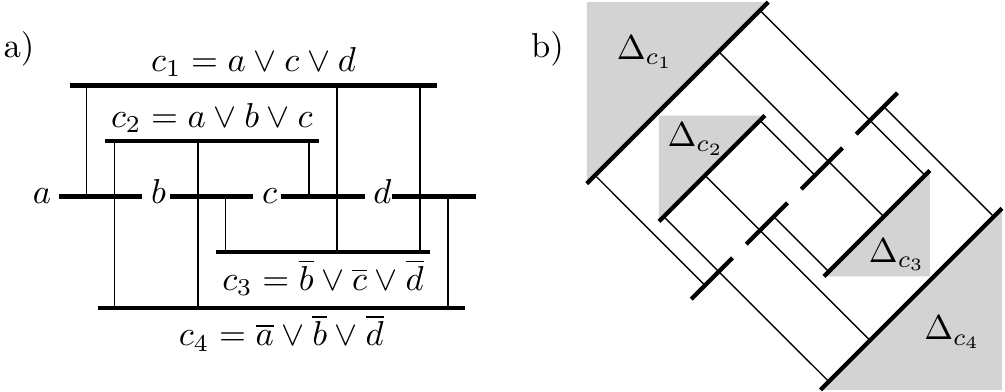}
    \caption{A representation of an instance of monotone planar 3-SAT
      with four variables $a,b,c,d$ and four clauses $c_1,c_2,c_3,c_4$
      (a). Image of the vertically-stretched version of (a) under the
      mapping~$\Phi$ (b).}
\label{fig:formula}
  \end{figure}

  A \emph{box} is an axis-aligned rectangle whose bottom-left
  and the top-right corners contain two $H$-vertices, connected 
  by a $G$-edge. We consider non-degenerate boxes, and thus this 
  $G$-edge requires at least one bend; when this edge is drawn with 
  one bend, there is a choice whether it contains the top-left or the 
  bottom-right corner of the box. In these cases we say that the box is 
  \emph{drawn top} and \emph{drawn bottom}, respectively. We now describe our
  \emph{variable},  \emph{pipe}, and  \emph{clause gadgets}.

  A \emph{variable gadget} consists of $h>0$ boxes $R_1,\dots,R_h$ that are
  $3 \times 3$-squares, where the bottom-left corner of $R_i$ lies at 
  $b + (2(i-1),2(i-1))$, for an arbitrary base
  point $b$; see Fig.~\ref{fig:variable}a-b.  The crucial property is that
  in a one-bend drawing of the gadget, $R_i$ is drawn bottom if
  and only if~$R_{i+1}$ is drawn top for $i=1,\dots,h-1$. Thus, in
  such a drawing, either all the \emph{odd boxes} (those with odd indices)
  are drawn top and all the \emph{even boxes} (those with even indices)
  are drawn bottom, or vice versa. This will be used to encode the 
  truth value of a variable.
  
  \begin{figure}[t]
    \centering
    \includegraphics[page=2]{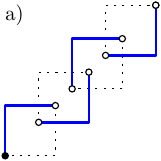} \hfil \includegraphics[page=3]{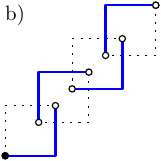} \hfil \hfil \includegraphics[page=5]{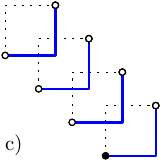} \hfil \includegraphics[page=6]{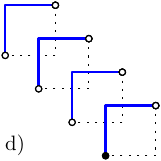}
    \caption{Variable gadget with $h=4$ boxes (a,b).  In (a) the even
      boxes are drawn top and the odd boxes are drawn bottom, (b)
      shows the opposite.  Pipe gadget (c,d).  In (c) all boxes are drawn bottom, in (d) they are all drawn top.  In all cases the base point is marked.}
    \label{fig:variable}
  \end{figure}

  A \emph{(positive) pipe gadget} works similarly; see
  Fig.~\ref{fig:variable}c-d. For a base point $b$,
  it consists of $h>0$ boxes $R_1,\dots,R_h$ that are
  $3\times3$-squares such that the bottom-left corner of $R_i$ 
  lies at $b+(-2(i-1),2(i-1))$; see Fig.~\ref{fig:variable}c-d.
  The decisive property is that in a one-bend drawing of the gadget, 
  all the boxes are drawn the same as $R_1$, that is, either all
  bottom (see Fig.~\ref{fig:variable}c) or all top (see
  Fig.~\ref{fig:variable}d). Negative pipe gadgets are symmetric with respect to
  the line $y=x$ and behave symmetrically.  
  
  The last gadget we describe is the \emph{(positive) clause gadget}; 
  negative clause gadgets are symmetric with respect to the line $y=x$
  and behave symmetrically.  
  The positive clause gadget has three
  \emph{input boxes} $R_1,R_2,R_3$, whose 
  corners lie on a single line with slope~1; we assume that $R_1$ lies left of
  $R_2$, which in turn lies left of $R_3$. To simplify the 
  description, we assume that the left lower corners of these
  rectangles lie at $(x,x),(y,y)$, and $(z,z)$, respectively.  Refer to Fig.~\ref{fig:clause}a.

  We create three \emph{literal} boxes $L_1,L_2,L_3$ that are
  $3\times3$-squares.  The lower left corner of $L_1$ is $(x-3,y+2)$,
  the lower left corner of $L_2$ if $(y-2,y+2)$, and the lower left
  corner of $L_3$ is $(y,z+3)$.  Note that the interiors of $L_2$
  and~$R_2$ intersect in a unit square, and therefore, if~$R_2$ is
  drawn top, then $L_2$ must be drawn top.  To obtain the same
  behavior for the other input and literal rectangles, we add two
  \emph{transmission boxes} $T_1$ and~$T_2$.  The lower left corner of
  $T_1$ is $(x-1,x+2)$ and its upper right corner is $(x+1,y+4)$.  The
  bottom-left and top-right corner of $T_2$ are $(y+2,z+2)$
  and~$(z+1,z+4)$, respectively. This guarantees that, also for $i=1,3$, 
  if $R_i$ is drawn top, then $T_i$ and $L_i$ are drawn top. We
  finally have a \emph{blocker} box $B$, with corners at~$(x-1,z+1)$ 
  and~$(x+1,z+4)$; and a \emph{clause box}, whose corners are in the centers 
  of $L_1$ and~$L_3$, respectively.
  
  Note that the $G$-edge connecting the two corners of the clause box, which 
  we call the \emph{clause edge}, requires at least two
  bends, as any one-bend drawing cuts horizontally through either the
  blocker $B$ or the literal square $L_2$; see Fig.~\ref{fig:clause}a.
  The following claim shows that the possibility of drawing it with 
  exactly two bends depends on the drawings of the literal boxes of
  the clause gadget, and thus on the truth values of the literals; see the
  Appendix and Fig.~\ref{fig:clause}b-c.

\spnewtheorem{myclaim}{Claim}{\bfseries}{\itshape}

\begin{restatable}[$\star$]{myclaim}{clause}
	\label{cl:clause}
If the other edges are drawn with one
bend, then the clause edge can be drawn with two bends if and only
if not all literal boxes are drawn top.
\end{restatable}

  \begin{figure}[t]
    \centering
    \includegraphics[page=10]{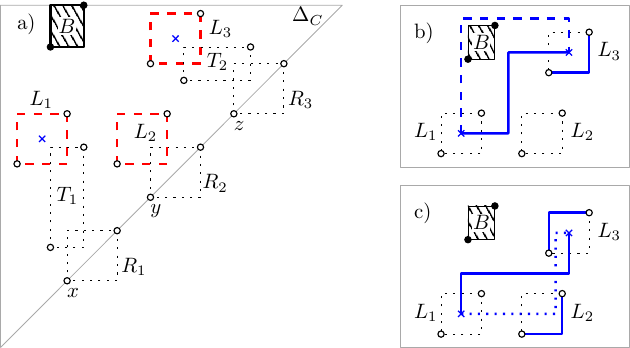}
    \caption{Clause gadget with input rectangles $R_1,R_2,R_3$.  The
      bottom-left and top-right corner of the clause box are drawn as
      crosses (a).  The image of the triangle~$\Delta_C$ under the
      mapping $(x,y) \mapsto (x-y,x+y)$ is drawn gray.  The
      possibilities of routing the clause edge with two bends, if
      $L_3$ is drawn bottom (b) and if $L_3$ is drawn top and $L_2$ is
      drawn bottom (c).}
    \label{fig:clause}
  \end{figure}

  We are now ready to put the construction together. Consider the layout
  of the variable--clause graph, where each
  variable $x$ is represented by a horizontal segment $s_x$ on the
  $x$-axis, and each clause $C=(c_1,c_2,c_3)$ with only positive (only negative) 
  literals by a horizontal segment $s_C$ above (below) the $x$-axis.  
  Further, the occurrence of a variable $x$ in a clause $C$ is represented by a vertical
  visibility segment $s_{x,C}$ that starts at an inner point of $s_x$
  and ends at an inner point of $s_C$; see Fig.~\ref{fig:formula}a.
  We call these points \emph{attachment points}.  By suitably
  stretching the drawing horizontally, we may assume that all segments
  start and end at points with integer coordinates divisible
  by~8.  We also stretch the whole construction vertically by a factor
  of $n$, which guarantees that for each clause segment
  $s_C$ the right-angled triangle $\Delta_C$, whose long side is $s_C$
  and that lies above $s_C$ (below $s_C$ if $C$ consists of negative
  literals) does not intersect any other segments in its interior.
  Note that the initial drawing fits on a grid of polynomial
  size~\cite{kr-pcr-92}, and the transformations only increase the area
  polynomially. For the construction it is useful to consider this
  representation rotated by $45^\circ$ in counterclockwise direction
  and scaled by a factor of $\sqrt{2}$ back to the grid.  This is
  achieved by the affine mapping~$\Phi \colon (x,y) \mapsto (x-y,x+y)$; see
  Fig.~\ref{fig:formula}b.

  For each variable segment $s_x$ with left endpoint $(a,0)$ and right
  endpoint $(b,0)$ we create a variable gadget with $h=(b-a)/2$ boxes
  and base point $(a,a)$.  For each clause segment $s_C$ above the
  $x$-axis with attachment points $(a_1,b),(a_2,b), (a_3,b)$, we create
  a positive clause gadget with input boxes at $(a_i-b,a_i+b)$.
  For each vertical segment $s_{x,C}$ above
  the $x$-axis with attachment points $(a,0)$ and $(a,b)$, we create a
  positive pipe gadget of $h=(b/2)-2$ boxes at base point $(a-2,a-2)$.
  Note that, together with the box of the variable gadget of $x$
  at $(a,a)$ and the input box of $C$ at $(a-b,a+b)$, the newly placed
  boxes form a pipe gadget that consists of $h+2$ boxes.  
  Since distinct vertical segments on the same side of
  the $x$-axis have horizontal distance at least~8, the boxes of
  distinct pipes do not intersect, and the placement is such that only
  the first and last box of each pipe gadget intersect boxes that
  belong to the corresponding variable or clause gadget.  Finally note
  that for each clause $C$, except for the input boxes, the clause
  gadget lies inside the image of the triangle~$\Delta_C$ under the
  mapping $\Phi$, since the attachment points are interior points of
  $s_C$, and the $x$-coordinates of its endpoints are divisible by~8.
  Hence, the only interaction of the clause gadget with the remainder
  of the construction is via the input variables
  The proof of the following claim, in the 
  Appendix, is based on showing that we can draw each box with exactly one
  bend and each clause edge with exactly two bends, if and only if the
  original instance of monotone planar 3-SAT is satisfiable.
%

\begin{restatable}[$\star$]{myclaim}{equivalence}
	\label{cl:equivalence}
Let $\varphi$ be an instance of monotone planar
3-SAT, with $\gamma$ clauses. Also, let $\beta$ be the number of boxes 
in the instance $(G,H,\Gamma_H)$ of \ortho constructed as described above.
Then, the formula~$\varphi$ is satisfiable if and only if the
instance $(G,H,\Gamma_H)$ admits an extension 
with at most $k=\beta+\gamma$ bends.
\end{restatable}

  Since the construction has polynomially many vertices and edges on a polynomial size grid, it
  can be executed in polynomial time. Moreover, by construction, 
  $V(H) = V(G)$, $E(H) = \emptyset$, and $E(G)$ is a
  matching. The statement of the theorem follows.
\end{proof}

By subdividing each non-clause edge with a $G$-vertex, and each clause
edge with two $G$-vertices, we get the following corollary.

\begin{corollary}
  It is NP-complete to decide whether a partial orthogonal drawing
  $(G,H,\Gamma_H)$ admits an extension without bends.
\end{corollary}

Similarly, we can ask whether an instance $(G,H,\Gamma_H)$ admits an
extension with at most $k$ bends per edge for a fixed number $k$.  The
construction depicted in Fig.~\ref{fig:modified-gadgets} shows how to
force an edge to use $k$ bends for any fixed number $k$.  By making
the part that enforces the first $k-1$ bends sufficiently small, we
essentially obtain the behavior of the box gadget from the proof of
Theorem~\ref{thm:bend-optimal-hardness}.

\begin{figure}
  \centering
  \includegraphics[page=1]{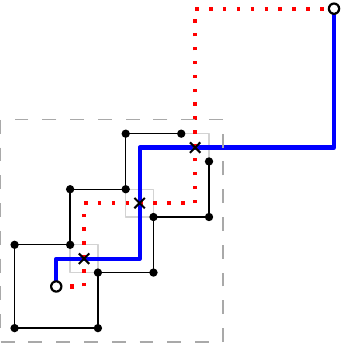}
  \caption{Gadget for forcing an edge to use $k=4$ bends.  All
    vertices and the thin solid black lines are $H$-vertices.  Up to
    minor geometric adjustments, the thick blue and dotted red lines
    show the only two ways to draw the $G$-edge between the two
    $H$-vertices $u$ and~$v$ with $k$ bends.  Scaling the
    lower left part to make it sufficiently small results in
    a construction that behaves like a box.}
  \label{fig:modified-gadgets}
\end{figure}

\begin{corollary}
  For any fixed $k \ge 2$, it is NP-complete to decide whether an instance 
  $(G,H,\Gamma_H)$ of \ortho admits an extension that
  uses at most $k$ bends per edge, even if $V(G) = V(H)$.
\end{corollary}

On the positive side, if all vertices are predrawn, the existence of
an extension with at most $k$ bends per edge can be tested efficiently
for $k=0$ and~$k=1$.

\begin{theorem}
  Let $(G,H,\Gamma_H)$ be an instance of \ortho with
  $V(G) = V(H)$ and let $k \in \{0,1\}$.  It can be tested in
  polynomial time whether $(G,H,\Gamma_H)$ admits an extension with at
  most $k$ bends per edge.
\end{theorem}

\begin{proof}
  For $k=0$ we simply draw each $G$-edge as the straight-line segment
  between its endpoints, and check whether this is a
  crossing-free orthogonal drawing.

  For $k=1$ we proceed as follows.  While there exists a $G$-edge
  $e=uv$ whose endpoints have the same $x$- or the same $y$-coordinates,
  we do the following.  If $e$ must be drawn as a straight-line (if
  $u$ and $v$ have the same $x$- or the same $y$-coordinates), the instance
  $(G,H,\Gamma_H)$ is equivalent to the instance $(G,H',\Gamma_H')$,
  where $H'$ is obtained from $H$ by adding $e$, and~$\Gamma_H'$ is
  obtained from inserting $e$ as a straight-line segment.  By applying
  this reduction rule, we eventually arrive at an instance
  $(G'',H'',\Gamma_H'')$ such that the endpoints of each $G$-edge have
  distinct $x$- and distinct $y$-coordinates.  Now for each such edge,
  there are precisely two ways to draw them with one bend. It is then
  straightforward to encode the existence of choices that lead to a
  planar drawing into a 2-SAT formula.
\end{proof}

\section{Conclusions}
\label{sec:conclusion}

In this paper we studied the problem of extending a partial
orthogonal drawing. We gave a linear-time algorithm to test the
existence of such an extension, and we proved that if one exists,
then there is also one whose edge complexity is linear in the size of
the given drawing. On the other hand, we showed that, if we also
restrict to a fixed constant the total number of bends or 
the number of bends per edge, then deciding the existence of an extension
is NP-hard.

Concerning future work we feel that the most important questions are
the following: 1) The complexity of $270 |V(H)|$ bends per edge
resulting from the transition to orthogonal drawings is significantly
worse than the one of $72|V(H)|$ bends per edge in the case of
arbitrary polygonal drawings~\cite{partial}.  Can this number 
be significantly reduced to, say, less than $100|V(H)|$?  2) As
mentioned in the introduction, Tamassia~\cite{t-emn-87} already
observed that an orthogonal representation of~$H$ can
be efficiently extended to an orthogonal representation of
$G$. However, drawing such an extension may require
to modify the drawing~$\Gamma_H$ of the given subgraph.  
Is it possible to efficiently test
whether a given orthogonal representation can be drawn such that it
extends a given drawing~$\Gamma_H$?

\bibliographystyle{splncs04}
\bibliography{ortho1}

\begin{thebibliography}{10}
\providecommand{\url}[1]{\texttt{#1}}
\providecommand{\urlprefix}{URL }
\providecommand{\doi}[1]{https://doi.org/#1}

\bibitem{accdd-sop-16}
Angelini, P., Chaplick, S., Cornelsen, S., {Da Lozzo}, G., {Di Battista}, G.,
  Eades, P., Kindermann, P., Kratochv{\'{\i}}l, J., Lipp, F., Rutter, I.:
  Simultaneous orthogonal planarity. In: Hu, Y., N{\"{o}}llenburg, M. (eds.)
  Proceedings of the 24th International Symposium on Graph Drawing and Network
  Visualization (GD'16). Lecture Notes in Computer Science, vol.~9801, pp.
  532--545. Springer (2016). \doi{10.1007/978-3-319-50106-2\_41}

\bibitem{adfjkpr-tppeg-15}
Angelini, P., {Di Battista}, G., Frati, F., Jel{\'{\i}}nek, V.,
  Kratochv{\'{\i}}l, J., Patrignani, M., Rutter, I.: Testing planarity of
  partially embedded graphs. {ACM} Trans. Algorithms  \textbf{11}(4),
  32:1--32:42 (2015). \doi{10.1145/2629341}

\bibitem{blr-ogdie-16}
Bl{\"{a}}sius, T., Lehmann, S., Rutter, I.: Orthogonal graph drawing with
  inflexible edges. Comput. Geom.  \textbf{55},  26--40 (2016).
  \doi{10.1016/j.comgeo.2016.03.001}

\bibitem{br-oogdc-16}
Bl{\"{a}}sius, T., Rutter, I., Wagner, D.: Optimal orthogonal graph drawing
  with convex bend costs. {ACM} Trans. Algorithms  \textbf{12}(3),  33:1--33:32
  (2016). \doi{10.1145/2838736}

\bibitem{partial}
Chan, T.M., Frati, F., Gutwenger, C., Lubiw, A., Mutzel, P., Schaefer, M.:
  Drawing partially embedded and simultaneously planar graphs. J. Graph
  Algorithms Appl.  \textbf{19}(2),  681--706 (2015). \doi{10.7155/jgaa.00375}

\bibitem{contactrep}
Chaplick, S., Dorbec, P., Kratochv{\'i}l, J., Montassier, M., Stacho, J.:
  Contact representations of planar graphs: Extending a partial representation
  is hard. In: Kratsch, D., Todinca, I. (eds.) Graph-Theoretic Concepts in
  Computer Science. pp. 139--151. Springer International Publishing, Cham
  (2014)

\bibitem{cfk-eprcg-19}
Chaplick, S., Fulek, R., Klav{\'{\i}}k, P.: Extending partial representations
  of circle graphs. Journal of Graph Theory  \textbf{91}(4),  365--394 (2019).
  \doi{10.1002/jgt.22436}

\bibitem{dk-obsps-12}
{de Berg}, M., Khosravi, A.: Optimal binary space partitions for segments in
  the place. International Journal of Computational Geometry \& Applications
  \textbf{22}(3),  187--205 (2012)

\bibitem{dlv-sood-98}
{Di Battista}, G., Liotta, G., Vargiu, F.: Spirality and optimal orthogonal
  drawings. {SIAM} J. Comput.  \textbf{27}(6),  1764--1811 (1998).
  \doi{10.1137/S0097539794262847}

\bibitem{dl-codve-98}
Didimo, W., Liotta, G.: Computing orthogonal drawings in a variable embedding
  setting. In: Chwa, K.Y., Ibarra, O.H. (eds.) Algorithms and Computation, 9th
  International Symposium, {ISAAC} '98, Taejon, Korea, December 14-16, 1998,
  Proceedings. Lecture Notes in Computer Science, vol.~1533, pp. 79--88.
  Springer (1998). \doi{10.1007/3-540-49381-6\_10}

\bibitem{dlop-oodp3-20}
Didimo, W., Liotta, G., Ortali, G., Patrignani, M.: Optimal orthogonal drawings
  of planar 3-graphs in linear time. In: Chawla, S. (ed.) Proceedings of the
  30th {ACM-SIAM} Symposium on Discrete Algorithms (SODA'20). pp. 806--825.
  {SIAM} (2020). \doi{10.1137/1.9781611975994.49}

\bibitem{dg-popda-13}
Duncan, C.A., Goodrich, M.T.: Planar orthogonal and polyline drawing
  algorithms. In: Tamassia, R. (ed.) Handbook on Graph Drawing and
  Visualization, pp. 223--246. Chapman and Hall/CRC (2013)

\bibitem{eghkn-ep1pd-20}
Eiben, E., Ganian, R., Hamm, T., Klute, F., N{\"o}llenburg, M.: {Extending
  Partial 1-Planar Drawings}. In: Czumaj, A., Dawar, A., Merelli, E. (eds.)
  47th International Colloquium on Automata, Languages, and Programming (ICALP
  2020). Leibniz International Proceedings in Informatics (LIPIcs), vol.~168,
  pp. 43:1--43:19. Schloss Dagstuhl--Leibniz-Zentrum f{\"u}r Informatik,
  Dagstuhl, Germany (2020). \doi{10.4230/LIPIcs.ICALP.2020.43},
  \url{https://drops.dagstuhl.de/opus/volltexte/2020/12450}

\bibitem{rect}
Eppstein, D.: Graph-theoretic solutions to computational geometry problems. In:
  Paul, C., Habib, M. (eds.) Graph-Theoretic Concepts in Computer Science. pp.
  1--16. Springer Berlin Heidelberg, Berlin, Heidelberg (2010)

\bibitem{kaufmann}
F{\"o}{\ss}meier, U., Kaufmann, M.: Drawing high degree graphs with low bend
  numbers. In: Brandenburg, F.J. (ed.) Graph Drawing. pp. 254--266. Springer
  Berlin Heidelberg, Berlin, Heidelberg (1996)

\bibitem{gt-ccurp-01}
Garg, A., Tamassia, R.: On the computational complexity of upward and
  rectilinear planarity testing. {SIAM} J. Comput.  \textbf{31}(2),  601--625
  (2001). \doi{10.1137/S0097539794277123}

\bibitem{jkr-kttpp-13}
Jel{\'{\i}}nek, V., Kratochv{\'{\i}}l, J., Rutter, I.: A {K}uratowski-type
  theorem for planarity of partially embedded graphs. Comput. Geom.
  \textbf{46}(4),  466--492 (2013). \doi{10.1016/j.comgeo.2012.07.005}

\bibitem{kkors-eprpu-17}
Klav{\'{\i}}k, P., Kratochv{\'{\i}}l, J., Otachi, Y., Rutter, I., Saitoh, T.,
  Saumell, M., Vyskocil, T.: Extending partial representations of proper and
  unit interval graphs. Algorithmica  \textbf{77}(4),  1071--1104 (2017).
  \doi{10.1007/s00453-016-0133-z}

\bibitem{kkos-eprsc-15}
Klav{\'{\i}}k, P., Kratochv{\'{\i}}l, J., Otachi, Y., Saitoh, T.: Extending
  partial representations of subclasses of chordal graphs. Theor. Comput. Sci.
  \textbf{576},  85--101 (2015). \doi{10.1016/j.tcs.2015.02.007}

\bibitem{kkosv-eprig-17}
Klav{\'{\i}}k, P., Kratochv{\'{\i}}l, J., Otachi, Y., Saitoh, T., Vyskocil, T.:
  Extending partial representations of interval graphs. Algorithmica
  \textbf{78}(3),  945--967 (2017). \doi{10.1007/s00453-016-0186-z}

\bibitem{kr-pcr-92}
Knuth, D.E., Raghunathan, A.: The problem of compatible representatives. {SIAM}
  J. Discret. Math.  \textbf{5}(3),  422--427 (1992). \doi{10.1137/0405033}

\bibitem{kw-eprtg-17}
Krawczyk, T., Walczak, B.: Extending partial representations of trapezoid
  graphs. In: Bodlaender, H.L., Woeginger, G.J. (eds.) Proceedings of the 43rd
  International Workshop on Graph-Theoretic Concepts in Computer Science
  (WG'17). Lecture Notes in Computer Science, vol. 10520, pp. 358--371.
  Springer (2017). \doi{10.1007/978-3-319-68705-6\_27}

\bibitem{pach}
Pach, J., Wenger, R.: Embedding planar graphs at fixed vertex locations. Graphs
  and Combinatorics  \textbf{17},  717--728 (2001)

\bibitem{p-epsld-06}
Patrignani, M.: On extending a partial straight-line drawing. Int. J. Found.
  Comput. Sci.  \textbf{17}(5),  1061--1070 (2006).
  \doi{10.1142/S0129054106004261}

\bibitem{t-emn-87}
Tamassia, R.: On embedding a graph in the grid with the minimum number of
  bends. Journal on Computing  \textbf{16}(3),  421--444 (1987)

\end{thebibliography}

\clearpage
\appendix

\section*{Appendix: Omitted Proofs}
\label{app:omitted}

\characterization*
\begin{proof}
	One direction is trivial; namely, if there exists an orthogonal drawing $\Gamma_G$ of $G$ that extends $\Gamma_H$, then the embedding of $G$ in $\Gamma_G$ satisfies the two properties by construction.
	Suppose now that there exists an embedding $\mathcal E_G$ of $G$ that satisfies the two properties. Since $\mathcal E_G$ is planar and extends $\mathcal E_H$, we can route each $G$-edge $uv$ as an arbitrary curve, while respecting the rotation at $u$ and $v$ in $\mathcal E_G$, without crossing any other edge. Also, the fact that $\mathcal E_G$ satisfies the port constraints in $\mathcal P_H$ implies that, for each $G$-edge $uv$, we can assign free ports of $u$ and $v$ to $uv$, in such a way that no port is assigned to more than one edge. Thus, by approximating the curve representing each $G$-edge $uv$ with an orthogonal polyline, it is possible to construct an orthogonal drawing of $G$ extending $\Gamma_H$. Note that in Theorem~\ref{orthoplanar} we will prove that this can even be done by using orthogonal polylines with a limited number of bends.
\end{proof}

\reduction*
\begin{proof}
	Suppose that $(G',H',\mathcal E_{H'})$ admits an embedding extension, and let $\mathcal E_{G'}$ be the corresponding embedding of $G'$.
	We construct an embedding $\mathcal E_{G}$ of $G$ that determines an embedding extension of $(G,H,\mathcal E_H,\mathcal P)$ that satisfies the port constraints, as follows. Let $v$ be any vertex of $G$. By construction, $v$ is also a vertex of $G'$. 
	
	Suppose first that all the neighbors of $v$ in $G'$ also belong to $G$. By construction, we have that in $(G,H,\mathcal E_H,\mathcal P)$ either there exists no $G$-edge incident to $v$, or there exists at most one $H$-edge incident to $v$. Also, in the former case, the rotation at $v$ in $\mathcal E_{H'}$ is the same as the one in $\mathcal E_{H}$, and there exists no port constraint at $v$. In the latter case, on the other hand, every rotation at $v$ in $\mathcal E_{G}$ trivially extends the rotation at $v$ in $\mathcal E_{H}$ and satisfies the (at most one) port constraint at $v$ in $\mathcal P$. Thus, in this case, we set the rotation at $v$ in $\mathcal E_{G}$ to be the same as the one in $\mathcal E_{G'}$.
	
	Suppose then that there exists exactly one neighbor of $v$ in $G'$ that does not belong to $G$. Then, by construction, this neighbor of $v$ is the vertex $w'$ that we introduced in one of the first two cases we described above. Namely, either it holds that $\deg_H(v) = 3$ and $\deg_G(v) = 4$, or it holds that $\deg_H(v) = 2$, $\deg_G(v) \geq 3$, $\mathcal P(e_1,e_2)=2$, and $\mathcal P(e_2,e_1)=0$, where $e_1$ and $e_2$ are the $H$-edges incident to $v$. In both cases, we obtain the rotation at $v$ in $\mathcal E_{G}$ by contracting the edge $vw'$, and by merging the rotations at $v$ and at $w'$ in $\mathcal E_{G'}$. This guarantees that the rotation at $v$ in $\mathcal E_{G}$ extends the rotation at $v$ in $\mathcal E_{H}$ and that the port constraint $\mathcal P(e_1,e_2)=1$ (resp.\ $\mathcal P(e_1,e_2)=2$) at $v$ is satisfied, since the edge $vw$ (resp.\ the edges $vw$ and $vz$) appear between $e_1$ and $e_2$ in $\mathcal E_{G}$.
	
	Suppose finally that there exists more than one neighbor of $v$ in $G'$ that does not belong to $G$. Then, by construction, $\deg_{G'}(v) = 4$, and the four neighbors of $v$ in $G'$ are the ones that we introduced in the last case we described above. Namely, $v$ is incident to two $H$-edges $e_1$ and $e_2$, and $\mathcal P(e_1,e_2)=\mathcal P(e_2,e_1)=1$. Observe that, since the subgraph of $G'$ induced by the vertices $v,w',z',w_1'$, and $w_2'$ is triconnected, the vertices $w',w_1',z',w_2'$ appear in the rotation at $v$ in $\mathcal E_{G'}$ either in this order or in its reverse. In the former case (the latter being analogous), we set the rotation at $v$ in $\mathcal E_{G}$ so that $w,w_1,z,w_2$ appear in this order. This trivially extends the rotation at $v$ in $\mathcal E_H$, since $\deg_H(v)=2$, and guarantees that the port constraints at $v$ are satisfied, since $w$ and $z$ use non-consecutive ports of $v$. 
	
	We further observe that, due to our transformation, the cycles in $H$ bijectively correspond to the cycles in $H'$, and that a vertex lies inside a cycle in $\mathcal{E}_H$ if and only if it lies inside the corresponding cycle in $\mathcal{E}_{H'}$. Together with the above discussion, this implies that $\mathcal E_{G}$ extends $\mathcal E_{H}$, since $\mathcal E_{G'}$ extends $\mathcal E_{H'}$. Finally, since $G'$ contains $G$ as a minor, the fact that $\mathcal E_{G'}$ is a planar embedding implies that $\mathcal E_G$ is a planar embedding, which concludes the proof of this direction.
	
	The proof for the other direction is analogous. In fact, given a planar embedding $\mathcal E_G$ of $G$ that is a solution for the instance $(G,H,\mathcal E_H,\mathcal P)$, we can construct a planar embedding $\mathcal E_{G'}$ of $G'$ that determines an embedding extension of $(G',H',\mathcal E_{H'})$, as follows.
	
	Let $v$ be any vertex of $G'$. If $v$ is also a vertex of $G$ and the all the neighbors of $v$ in $G'$ also belong to $G$, then we can set the rotation at $v$ in $\mathcal E_{G'}$ to be the same as the one $\mathcal E_{G}$, as discussed above. To cover all the other cases (either $v$ or at least one of its neighbors is not a vertex of $G$), it is enough to consider the three cases in the construction we described above.
	
	In the first two cases, the fact that $\mathcal E_{G}$ satisfies the port constraint $\mathcal P(e_1,e_2)=1$ (resp. $\mathcal P(e_1,e_2)=2$) implies that $vw$ (resp.\ both $vw$ and $vz$) appears between $e_1$ and $e_2$ in the rotation at $v$ in $\mathcal E_{G}$. Thus, inserting $vw'$ in the rotation at $v$ in $\mathcal E_{G'}$ in the same position as $vw$ (resp.\ both $vw$ and $vz$) in the rotation at $v$ in $\mathcal E_{G}$ yields a rotation at $v$ in $\mathcal E_{G'}$ that extends the one at $v$ in in $\mathcal E_{H'}$. The same trivially holds for the rotation at $w'$, since $\deg_{H'}(w')=1$.
	
	In the last case, when $v$ is incident to two $H$-edges $e_1$ and $e_2$, and $\mathcal P(e_1,e_2)=\mathcal P(e_2,e_1)=1$, the fact that $\mathcal E_{G}$ satisfies the port constraints implies that the vertices $w,w_1,z,w_2$ appear in the rotation at $v$ in $\mathcal E_{G}$ either in this order or in its reverse. In both cases, it is possible to set the rotations at $v$, $w',w_1',z',w_2'$ in $\mathcal E_{G'}$ so that the triconnected subgraph induced by these vertices is embedded according to its unique planar embedding, and all the vertices of $G'$, except for $v$, lie outside of the cycle induced by $w',w_1',z',w_2'$. Note that each of these five vertices is incident to at most one $H'$-edge, and thus every of its rotations in $\mathcal E_{G'}$ trivially extends the one in $\mathcal E_{H'}$. This concludes the proof of the lemma.
\end{proof}
\subsection*{Complete Proof of Theorem~\ref{mainthm}}
\drawing*
To prove Theorem \ref{mainthm}, we first need a couple of tools and we present those tools as lemmas before delving into the actual proof of the theorem.
\par Since the embedding of $G$ is fixed, it is enough to consider a face $F$ of $\Gamma_{H}$ and prove Theorem \ref{mainthm} for that particular face. We first show how to construct  an inner $\eps$-approximating orthogonal polygon for each facial walk of $F$ using a technique similar to the on from~\cite{partial}.

\begin{lemma}
  Let $W$ be a facial walk in a face $F$ of an orthogonal drawing of a
  graph $G$ in the plane. A inner $\varepsilon$-approximating
  orthogonal polygon $P_\varepsilon$ of $W$ can be constructed in $O(|W|)$ time so that
  $P_\varepsilon$ has at most $\max \{4,|W|+l\}$ vertices, where $l$
  is the number of degree-1 vertices in $W$.  
\label{nest-approx}
\end{lemma}
\begin{proof}
  If $W$ is an isolated vertex $v$, then approximate $W$ with a square of sidelength~$2 \eps$ centered at $v$. Next, assume that $W$ contains more than one vertex. 
  We consider each vertex of degree~1 in $W$ as a sequence of two degree-2 vertices that are connected by an infinitely short edge that forms a~$270^\circ$-angle with the single edge incident to $v$ inside $F$. 
  Consider a corner $e,v,f$ of $W$, where $e$ and $f$ are two consecutive edges and $v$ is their shared vertex.  Let~$\alpha$ denote the angle formed by~$e$ and~$f$ inside~$F$.  If~$\alpha$ is a $180^\circ$-angle, then we choose~$v'$ as the point on the angular bisector of~$\alpha$ at distance~$\eps$ from~$v$.  Otherwise, we choose~$v'$ as the point on the angular bisector of~$\alpha$ at distance~$\sqrt{2}\eps$ from~$v$.
\begin{figure}[tb]
\begin{center}
\includegraphics{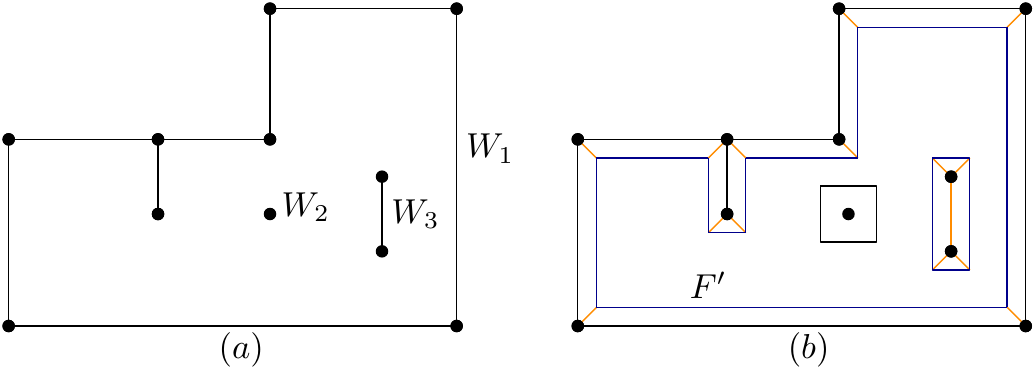}
\caption{(a) A face with outer walk $W_1$ and, inner facial walks $W_2$ and $W_3$. (b) An approximation $F^\prime$ of $F$.}
\label{fig:epsilon}
\end{center}
\end{figure}
If $(v_i)_{i=1}^{k}$ is the sequence of vertices in $W$, then by joining $(v_i^\prime)_{i=1}^{k}$, we get an orthogonal polygon that ~$\eps$-approximates $W$.
\end{proof}

We now prove two auxiliary lemmas, which follow the structure of Lemmas $5$ and $6$ in \cite{partial}. Assume that $G$ is a Hamiltonian graph with Hamiltonian cycle $C$. Lemma \ref{hamcyc} provides a method to draw the edges of $C$, assuming that the vertex locations are fixed. Lemma \ref{hamgraph} explains how to draw the remaining edges of $G$. 
\begin{lemma}
 Let $C$ be a cycle with fixed vertex locations, and suppose we are given an orthogonal planar drawing of a tree $T$, in which the vertices of $C$ are leaves of $T$ at their fixed locations and each edge of $T$ has at most $k$ bends. Then for every $\varepsilon >0$ there is a planar Kandinsky drawing of $C$ with at most $3k|E(T)|$ bends per edge and $\varepsilon$-close to $T$.
 \label{hamcyc}
\end{lemma}
\begin{proof}
Let $p_1,\ldots ,p_n$ be the vertices of the cycle $C$ in order. To construct a  planar poly-line drawing of $C$, Lemma 5 of \cite{partial} explains a method as follows. First of all, $n$ $\varepsilon$-approximations $\theta_i(1\leq i \leq n)$ of the given drawing of $T$ are 
constructed, using Lemma~\ref{nest-approx}. Then another poly-line polygonal curve $\theta_i^\prime$ is constructed from $\theta_i$ by ignoring the parts of $\theta_i$ corresponding to the vertices $v_{i+1},\ldots ,v_n$. The edge $p_ip_{i+1}$ is routed through $\theta_i^\prime$. 
In order to draw the edges of $C$, we follow the same method explained above by constructing $(i\varepsilon/n+1)$-approximation $\theta_i$ of the given orthogonal drawing of $T$ using Lemma \ref{nest-approx}, for $1\leq i\leq n$, and by routing the edges of $C$ through corresponding $\theta_i^\prime$'s. An example is given in Fig.~\ref{lemma2fig}.

\begin{figure}[h]
\begin{center}
\includegraphics{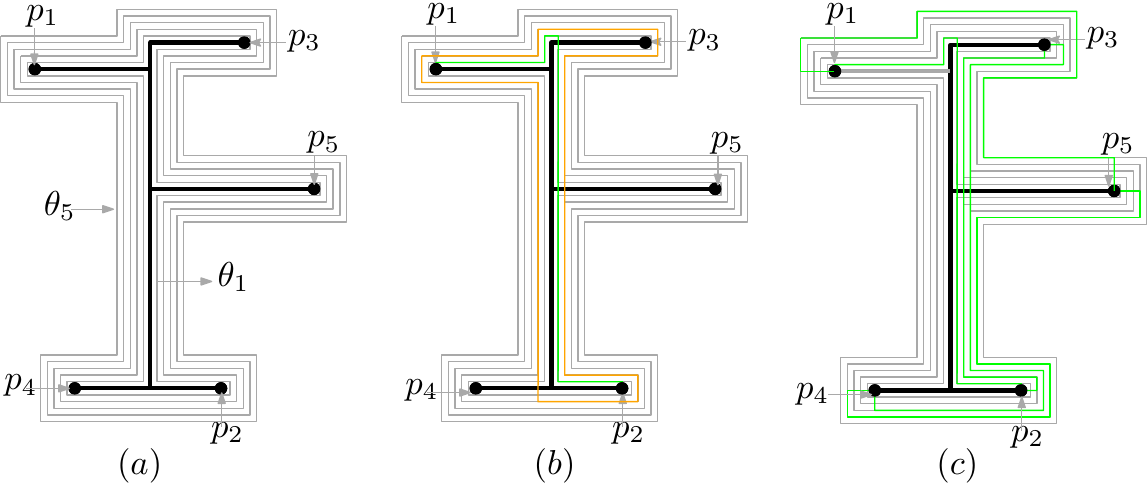}
\caption{(a) An orthogonal drawing of a tree $T$ together with approximations $\theta_i$. (b) The edge $p_1p_2$ is drawn in green color and $\theta_2^\prime$ is drawn in orange color. (c) Complete orthogonal drawing of $C$ with respect to $T$ }
\label{lemma2fig}
\end{center}
\end{figure}
Here, note that each edge of $C$ is replaced with a part of an approximation of  $\theta_i$ and  $\theta_i$ has at most $3k|E(T)|$ edges. Hence each edge of $C$ is replaced with an orthogonal arc that has at most $3k|E(T)|$ bends.
\end{proof}

\begin{lemma}
Let $G$ be a Hamiltonian multigraph with a given planar embedding and fixed vertex locations. Suppose we are given an orthogonal drawing of a tree $T$ whose leaves include all the vertices of $G$ at their fixed locations and each edge of $T$ has at most $k$ bends. Then for every $\varepsilon > 0$ there is a planar Kandinsky drawing of $G$ so that
\begin{enumerate} 
\item the drawing is $\varepsilon$-close to T,
\item the drawing realizes the given embedding,
\item the vertices of $G$ are at their fixed locations, 
\item every edge has at most $6k|E(T)|$ bends, and 
\item every edge comes close to any leaf of $T$ at most twice, and only does so by terminating at or bending near the leaf.
\end{enumerate}
\label{hamgraph}
\end{lemma}
\begin{proof}
We closely follow Lemma 6 of \cite{partial} to construct a planar poly-line drawing of $G$ that works as follows. Using Lemma 5 of \cite{partial}, a planar poly-line drawing of $C$ with respect to the given drawing of $T$ is constructed. Next, $m$ approximations $\Delta_{i,k}$ of $\theta_i^\prime$ are constructed for each $1\leq i\leq n$ and $1\leq k \leq m$, where $m=|E(G)|$. To route an edge $p_ip_j$, the path concatenating the straight-line polygons $\Delta_{i,k}$ and $\Delta_{j,k}$ is used. To construct a planar Kandinsky drawing of $G$, we continue in a similar manner. First, we route the edges of the Hamiltonian cycle $C$ using Lemma \ref{hamcyc} and then route the remaining edges by creating additional approximations of the curves $\theta_i^\prime$. Here, corresponding to an edge at most $k\times 6(E(T))=6kE(T)$ bends are introduced, since an edge $p_ip_j$ is a concatenation of two approximations $\Delta_{i,k}$ and $\Delta_{j,k}$.  An example is illustrated in Fig.~\ref{lemma3fig}.
\begin{figure}[h]
\begin{center}
\includegraphics{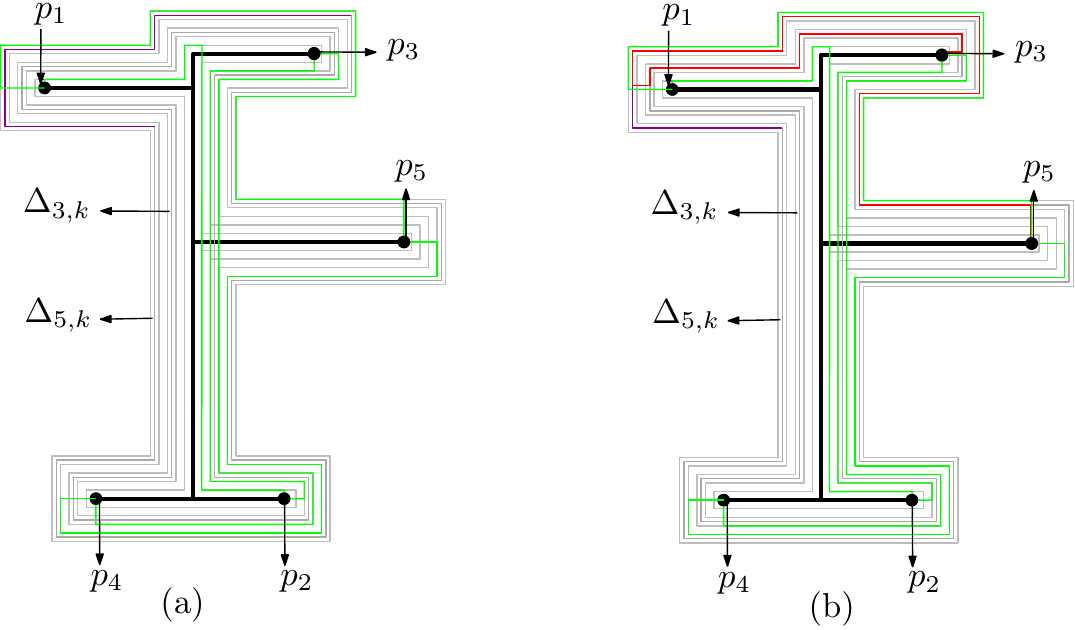}
\caption{The edge $p_3p_5$ is drawn. (a) The polygons $\Delta_{3,k}$ and $\Delta_{5,k}$ are drawn in gray color. (b) The edge $p_3p_5$ is drawn in red color using parts of $\Delta_{3,k}$ and $\Delta_{5,k}$.  }
\label{lemma3fig}
\end{center}
\end{figure}
\end{proof}

Now, in order to make the given graph Hamiltonian, we use the following result by Pach and Wenger \cite{pach}.

\begin{lemma}[Pach, Wenger \cite{pach}] For a planar graph $G$ a Hamiltonian planar graph $G'$ with  $|E(G^\prime)|\leq 5|E(G)|-10$ can be constructed from $G$ by subdividing and adding edges in linear time.  The construction is such that each edge of $G$ is subdivided by at most two new vertices.
 \label{makeham}
 \end{lemma}
Next, we assume that a planar embedding of the graph $G$ together with a set of vertices $U\subseteq V(G)$ is given, where every element of $U$ has a fixed location. The next lemma shows a method to route the edges of $G$ by converting it into a Hamiltonian graph and then contracting the edges if  at least one of its end point is not $U$. Finally, we undo the edge contractions to obtain a drawing of the original graph $G$. 

\begin{lemma}
Let $G$ be a multigraph with a given planar embedding and fixed locations for a subset $U$ of its vertices. Suppose we are given an orthogonal drawing of a tree $T$ whose leaves include all the vertices in $U$ at their fixed locations and each edge of $T$ has at most $k$ bends. Then for every $\varepsilon > 0$ there is a planar Kandinsky drawing of $G$ so that
\begin{enumerate}
\item the drawing is $\varepsilon$-close to $T$,
\item the drawing realizes the given embedding,
\item the vertices in $U$ are at their fixed locations, and
\item  each edge has at most $18k|V(T)|$ bends and comes close to each vertex $u$ in $U$ at most $6$ times, where coming close to u means intersecting an $\varepsilon$-neighborhood of $u$. Furthermore, any edge that comes close to $u$ will either terminate at $u$ or enter the $\varepsilon$-neighborhood of $u$, bend at a point in this $\varepsilon$-neighborhood, and then leave it.
\end{enumerate}\label{lemma3}
\end{lemma}

\begin{proof}
From a given graph $G$, construct a Hamiltonian graph $G^\prime$ with a Hamiltonian cycle $C$ by subdividing each edge of $G $ at most twice, and by adding some edges using Lemma \ref{makeham}. 
We traverse through $C$ and whenever we encounter an edge $e$ that has at least one endpoint not in $U$, then we contract $e$. Continue this process to get a multigraph $G^{\prime \prime}$ with a Hamiltonian cycle $C^\prime$ such that $V(G^{\prime \prime})=U$. 
Now, using Lemma \ref{hamgraph}, find an orthogonal drawing $\Gamma''$ for $G''$ with respect to~$T$. Fix a vertex $u\in V(G^{\prime \prime})$ and let $V_u$ be the vertices of $G$ that have been contracted into $u$. Next, we have to reconstruct the subgraph $G_u=G^\prime[V_u]$ and route the edges that connect vertices from $V_u$ to $V(G^{\prime })\setminus V_u$. To reconstruct $G_u$, construct a small disk around $u$ in $\Gamma''$. Since $\Gamma''$ is an orthogonal drawing, we can cover $N_{G^{\prime \prime}}(u)$  into four sets $(V_{u_i})_{i=1}^{4}$ depending on the side of $v$ to which its edge attaches.
Now, let $G^{\prime}_u=(V,E)$ with $V=V_u\cup N_{G^{\prime \prime}}(u)\cup\{u_1,u_2,u_3,u_4\}$ and $E=E(G_u)\cup\{u_1u_2, u_2u_3,u_3u_4,u_4u_1\}\cup\{u_ix: x\in V_{u_i} \text{ for } 1\leq i \leq 4\}\cup \{yu_i:\text{ there exists an edge } yx \text{ in }G^\prime \text{ with }x\in V_{u_i} \text{ and }y\in V_{u}\}$. 
\begin{figure}[tb]
\begin{center}
\includegraphics{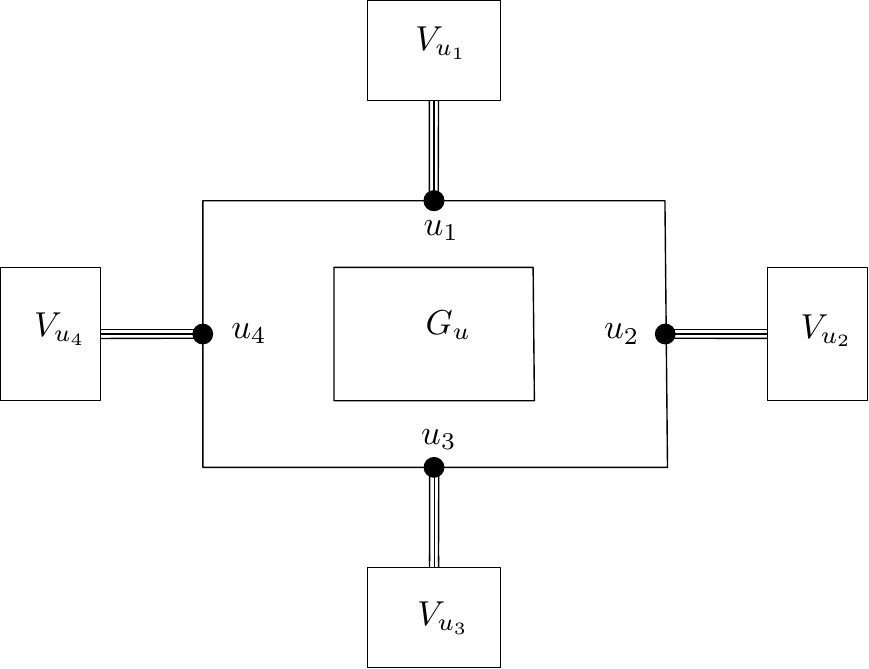}
\caption{The graph $G^{\prime}_u$.}
\end{center}
\end{figure}

Note that $G^{\prime}_u$ is a planar multigraph and hence it has a Kandinsky drawing $\Gamma_u$ with at most two bends on each edge that can be computed in linear time~\cite{kaufmann}. Using $ \Gamma_u$ we can route the edges that connect $V_u$ and $V_{u_i}$ by ignoring the vertex $u_i$. Thus we get a Kandinsky drawing of $G^\prime$ with at most $6|E(T)|+2$ bends per edge (using Lemma \ref{hamgraph} and the two extra bends that are added while reconstructing $G_u$). Since each edge of $G$ is subdivided at most twice to get $G^\prime$, each edge of $G$ has $3(6k|E(T)|+2)=18k|E(T)|+6<18k|V(T)|$ bends. 
In addition, since each edge of $G^\prime$ comes close to a leaf of $T$ at most twice, an edge of $G$ comes close to a vertex of $U$ at most six times.
\end{proof}
Now, we have all the required tools to prove Theorem \ref{mainthm}.
\subsubsection{Proof of Theorem \ref{mainthm}:}

Let $F$ be a face of $\Gamma_{H}$. Let $W_i: 1\leq i \leq a$ be facial walks inside $F$ with isolated vertices and let $W_i: a+1\leq i \leq a+b$ be facial walks inside $F$ that involve more than one vertex. 
Construct an inner $\eps$-approximation $F'$ of $F$ using Lemma \ref{nest-approx}. Let $W_i^\prime$ be the orthogonal polygon that approximates $W_i:1 \leq i \leq a+b$. Since $|W_i^\prime|\leq \max \{4, |W_i|+l\}$, we have $|F^\prime|\leq \sum\limits^{a+b}_{i=1}|W_i|+a+b$.  Now, partition $F^\prime$ into rectangles using at most $n/2+h-1$ rectangles in time $O(n^{3/2}\log n) $ \cite{rect}, where $n$ is the number of vertices and $h$ is the number of holes. So in our case, the number of rectangles will be $\frac{|F^\prime|}{2}+a-1\leq \frac{\sum\limits^{a+b}_{i=1}|W_i|+a+b}{2}+a-1\leq \frac{\sum\limits^{a+b}_{i=1}|W_i|}{2}+\frac{3a+b}{2}-1$. Place a vertex at the center of each rectangle. Construct a graph $K$ by joining the vertices of adjacent rectangles (we call two rectangles adjacent if they share one side) if the line segment joining the respective centers lies inside $F^\prime$ and the edge joining them has at most one bend. Note that $G^\prime$ is a connected graph. Let $T$ be a spanning tree of $G^\prime$. Then $T$ has $\frac{\sum\limits^{a+b}_{i=1}|W_i|}{2}+\frac{3a+b}{2}-1$ vertices and each edge of $T$ has at most four bends. Now, for each facial walk $W_i: 1\leq i\leq a$, add the corresponding isolated vertex as a leaf to $T$. For each facial walk $W_i: a+1\leq i\leq a+b$, add a new vertex near to $W_i$ as a leaf of $T$. This adds $a+b$ vertices to $T$ and now, the number of vertices of $T$ is $|V(T)|=a+b+\frac{\sum\limits^{a+b}_{i=1}|W_i|}{2}+\frac{3a+b}{2}-1=\frac{\sum\limits^{a+b}_{i=1}|W_i|}{2}+\frac{5a+3b}{2}-1$. 
\par Construct the multigraph $G_F$ induced by the vertices lying inside or on the boundary of $F$, by contracting each facial walk of $F$ to a single vertex. Draw $G_F$ along $T$ using Lemma \ref{lemma3}. Note that the vertices corresponding to facial walks (inside $F$) are drawn at fixed locations. Here, each edge of $G_F$ has at most $18k|V(T)|=18*4*|V(T)|=72$ bends. 

Now, we reconstruct the edges between $G_F$ and the non-isolated boundary components of $F$, following the same method as in Theorem 1, \cite{partial}. That is, by creating a buffer zone in between $F^\prime$ and $F$, the above mentioned edges are routed through the zone. This adds at most $|W_i|+5$ bends for each edge. 

Next, we have to add the edges of $G\setminus H$ that belong to $F$ according to the given embedding $\mathcal E_G$ of $G$. By Lemma \ref{hamgraph}, an edge can come close at most six times to a vertex in $U$ and thus an edge needs at most $6(|W_i|+5)=6|W_i|+30$ bends to go around $W_i$. So altogether there are at most $6\sum\limits_{i=a+1}^{a+b}|W_i|+30b$ bends along the whole edge to go around all the $W_i'$s. Since we started with $18\times 4|V(T)|=72|V(T)|$ bends (Lemma \ref{hamgraph}) for each edge, this number increased to at most $6\sum\limits_{i=a+1}^{a+b}|W_i|+30b+72|V(T)|$. Thus the total number of bends per edge can be calculated as follows.

\begin{align*}
6\sum\limits_{i=a+1}^{a+b}|W_i|+30b+72|V(T)|&\leq 6\sum\limits_{i=a+1}^{a+b}|W_i|+30b+72\left(\frac{\sum\limits^{a+b}_{i=1}|W_i|}{2}+\frac{5a+3b}{2}-1\right)\\
&\leq 41\sum\limits_{i=a+1}^{a+b}|W_i|+180a+138b-72\\
 &\hspace{24ex}\text{since  }|W_i|=1\text{ for }1 \leq i \leq a\\        
&\leq 41\times 2|V(H)|+180\times|V(H)|\text{ since } \\
&\hspace{5ex}\sum\limits_{i=a+1}^{a+b}|W_i|\leq 2|V(H)| \text{ and } a+2b\leq |V(H)| \text{\cite{partial}}\\  
&\leq 262|V(H)|
\end{align*}

\orthodrawing*

\begin{proof}
We first create a planar Kandinsky drawing $\Gamma^K_G$ of the given graph $G$ using Theorem \ref{mainthm}. Let $u$ be a vertex of $G$. Since $G$ has an orthogonal drawing, we have that $\deg(u)\leq 4$. 
Note that, in $\Gamma^K_G$, some of the edges incident to $u$ may be attached to the same port. Our goal is to change the port to which some of the edges are attached, in such a way that every edge is attached to a different port, while respecting the rotation at $u$ in $\mathcal{E}_G$.  Note that we only reroute $G$-edges, as $H$-edges have a fixed drawing and can therefore no two $H$-edges can attach to the same port of a vertex.  Since~$\Gamma_G$ is an orthogonal drawing extension,~$\mathcal E_G$ satisfies the port constraints, and such a rerouting can be achieved as illustrated in Fig.~\ref{fig:orthplnr}.  Note that this adds at most four bends per edge.

Applying this operation to all the vertices of $G$ yields a planar orthogonal drawing $\Gamma_G$ of $G$ that realizes $\mathcal{E}_G$, extends $\mathcal{H}$, and has at most $270|V(H)|$ bends per edge (at most twice four additional bends on each edge). 
\end{proof}

\subsection*{Complete Proofs for the Claims in Theorem~\ref{thm:bend-optimal-hardness}}

\clause*

\begin{proof}
	Suppose, for a contradiction, that the clause edge is drawn
	with two bends, but all three literal boxes are drawn top.  Then,
	starting from the center of $L_1$, the clause edge must first
	intersect the bottom or the right side of $L_1$.  If it intersects
	the bottom side, then it further consists of a horizontal segment
	and a vertical segment that then ends at the center of $L_3$.  But
	then either the horizontal segment cuts horizontally through $T_1$,
	or the vertical segment cuts vertically through $R_2$.  Both cases
	contradict the assumption that the drawing is without crossings.
	Hence we can assume that the clause edge intersects the right side
	of~$L_1$.  Since it cannot intersect the left side of $L_2$, there
	must be a bend on the segment between the centers of $L_1$ and~$L_2$
	that lies outside of these two boxes.  The rest of the clause edge
	is then drawn from this bend with one additional bend to $L_3$.
	However, then this part of the edge either cuts horizontally through
	$L_2$, or it intersects the left side of~$L_3$; in either case, the
	edge has a crossing.
	
	On the other hand, we show that if at least one of $L_1,L_2,L_3$ is
	not drawn top, then we can draw the clause edge with two bends.
	Assume that $L_3$ is drawn bottom.  Depending on whether the
	top-left or bottom-right corner of $L_1$ is used, we can draw the
	clause edge as indicated by the solid or the dashed curve in
	Fig.~\ref{fig:clause}b.  Note that this is independent of whether
	$L_2$ is drawn top or bottom.  Now assume that $L_3$ uses its
	top-left corner.  If~$L_1$ is drawn bottom, we can draw the clause
	edge as indicated by the solid curve in Fig.~\ref{fig:clause}c.
	Finally, if both $L_1$ and~$L_3$ use their top-left corner, but
	$L_2$ does not, we can route the clause edge as indicated by the
	dotted curve in Fig.~\ref{fig:clause}c.
\end{proof}

\equivalence*

\begin{proof}
	Assume we are given a satisfying assignment of~$\varphi$.  For each
	variable, we draw the odd boxes bottom and the even boxes top if the
	variable is assigned value true, and the other way around if it is
	false.  For each clause $C$, let $x$ be a variable that satisfies
	it.  We discuss the case that $C$ contains only positive literals,
	the case that it only contains negative literals is symmetric.  We
	draw $C$ in such a way that the input box of $x$ is drawn bottom and
	all other input boxes are drawn top.  We draw the boxes of the pipe
	gadget that connects $x$ to $C$ bottom, and the remaining pipe
	gadgets that connect to other variables to $C$ top.  Note that the
	latter cannot cause crossings, and the former do not cause a
	crossing, since it only intersects with an odd box of the variable
	gadget of $x$, which is drawn bottom since $x$ is true.  As observed
	above, the clause edge of $C$ can be drawn with two bends.
	Altogether, we obtain a crossing-free orthogonal drawing~$\Gamma_G$ of
	the instance that has $\beta+\gamma$ bends (one bend per box, and one
	additional bend per clause).
	
	Conversely, assume that there exists a drawing~$\Gamma_G$ with $\beta+\gamma$
	bends.  Recall that each box requires at least one bend, and each
	clause edge requires at least two bends.  It follows that each
	clause edge is drawn with two bends, and that each edge of the
	remaining edges is drawn with one bend.  We now assign a variable
	$x$ the value true if and only if its odd boxes are drawn bottom.
	Let $C$ be a clause with only positive literals; the case with only
	negative literals is symmetric.  Since the clause edge of $C$ is
	drawn with two bends, it follows that at least one of the input
	boxes is drawn bottom.  Then all boxes of the corresponding pipe are
	also drawn bottom, and therefore an odd box of the corresponding
	variable is also drawn bottom.  Hence the variable is true and $C$
	is satisfied.
\end{proof}

\end{document}